\documentclass{fundam}

\usepackage{url} %
\usepackage{amssymb}
\usepackage{amsmath} 
\usepackage{graphicx}
\usepackage{enumerate}
\usepackage{tikz}
\usepackage{mathrsfs}
\usepackage{verbatim}
\usepackage{upgreek}
 \usepackage{mathptmx}
\usepackage{mathptmx}
\usepackage{eucal}
\usepackage{colortbl}

\newcommand{\n}{\operatorname{n}}
\newcommand{\epn}{\operatorname{epn}}

\usepackage{hyperref}

\begin{document}

\setcounter{page}{201}
\publyear{22}
\papernumber{2108}
\volume{185}
\issue{3}

  \finalVersionForARXIV

\title{Perfect Domination, Roman Domination and perfect Roman Domination in Lexicographic Product Graphs}

\author{A. Cabrera Mart\'inez\thanks{Address for correspondence:  Universitat Rovira i Virgili,
                      Departament d'Enginyeria Inform\`atica i Matem\`atiques,  Av. Pa\"{\i}sos Catalans 26,
                      43007 Tarragona, Spain. \newline \newline
          \vspace*{-6mm}{\scriptsize{Received January 2021; \ accepted March 2022.}}}, C. Garc\'{i}a-G\'{o}mez, J. A. Rodr\'{\i}guez-Vel\'{a}zquez
\\
 Universitat Rovira i Virgili \\
 Departament d'Enginyeria Inform\`atica i Matem\`atiques  \\
 Av. Pa\"{\i}sos Catalans 26, 43007 Tarragona, Spain \\
  \{abel.cabrera, carlos.garciag, juanalberto.rodriguez\}@urv.cat
  }

\maketitle

\runninghead{A. Cabrera Mart\'inez et al.}{Perfect Domination, Roman Domination and perfect Roman Domination...}

\vspace*{-4mm}
\begin{abstract}
The aim of this paper is to  obtain closed formulas for the perfect domination number,  the Roman domination number and the perfect Roman domination number of lexicographic product graphs.
We show that these formulas can be obtained relatively easily for the case of the first two parameters.  The picture is quite different when it concerns the perfect Roman domination number. In this case, we
obtain  general bounds and then we give sufficient and/or necessary conditions for the bounds to be achieved.
 We also discuss the case of perfect Roman graphs and we characterize the lexicographic product graphs where the perfect Roman domination number equals the Roman domination number.
\end{abstract}


\section{Introduction}
Given a graph $G$, a set $S\subseteq V(G)$ of vertices is a \emph{dominating set}  if
every vertex  in $V(G)\setminus S$ is adjacent to at least one vertex in $S$.
Let $\mathcal{D}(G)$ be the set of dominating sets of $G$.
The \emph{domination number} of $G$ is defined to be,
$$
\gamma(G)=\min\{|S|:\, S\in \mathcal{D}(G)\}.\vspace*{3mm}
$$
Now, $S\subseteq V(G)$ is a \emph{perfect dominating set} of $G$ if
every vertex  in $V(G)\setminus S$ is adjacent to exactly one vertex in $S$.
Let $\mathcal{D}^p(G)$ be the set of perfect dominating sets of $G$.
The \emph{perfect domination number} of $G$ is defined to be,
$$\gamma^p(G)=\min\{|S|:\, S\in \mathcal{D}^p(G)\}.$$
 Notice that  $\mathcal{D}^p(G)\subseteq \mathcal{D}(G)$, which implies that $\gamma(G)\le \gamma^{p}(G)$.

\medskip
The domination number  has been extensively studied. For instance, we cite the following books,  \cite{Haynes1998a,Haynes1998}. The theory of perfect domination was introduced
by Livingston and Stout in \cite{PerfectDomination} and has been studied by several authors, including \cite{MR1238868, Cockayne1993, MR2445802, MR1122221,MR3334589, MR3128528}.

\medskip
 Cockayne, et al.\ \cite{Cockayne2004} defined a {\it Roman dominating  function}, abbreviated RDF, on a graph $G$ to be a function $f: V(G)\longrightarrow \{0,1,2\}$ satisfying the condition that every vertex $u$ for which $f(u)=0$ is adjacent to at least one vertex $v$ for which $f(v)=2$. The \textit{weight} of $f$ is defined to be
 $$\omega(f)=\sum_{v\in V(G)}f(v).$$ For  $X\subseteq V(G)$ we define the weight of $X$  as  $f(X)= \sum_{v\in X}f(v)$. The \textit{Roman domination number}, denoted by   $\gamma_R(G)$, is the minimum weight among all Roman dominating functions on $G$, \textit{i.e.}, $$\gamma_R(G)=\min\{\omega(f):\, f \text{ is an RDF on } G\}.$$ An RDF of weight $\gamma_R(G)$ is called a $\gamma_R(G)$-function. Obviously,  $\gamma_R(G)\le 2\gamma(G)$ for every graph $G$. A \emph{Roman graph} is a graph $G$ with $\gamma_R(G)=2\gamma(G)$.

\medskip
Recently, a perfect version of Roman domination was introduced  by Henning, Klostermeyer and  MacGillivray  \cite{HENNING2018235}. They defined a \emph{perfect Roman dominating function}, abbreviated PRDF, as an RDF $f$ satisfying the condition that every vertex $u$ for which $f(u)=0$ is adjacent to exactly one vertex $v$ for which $f(v)=2$.
   The \textit{perfect Roman domination number}, denoted by   $\gamma_R^p(G)$, is the minimum weight among all perfect Roman dominating functions on $G$, \textit{i.e.}, $$\gamma_R^p(G)=\min\{\omega(f):\, f \text{ is a PRDF on } G\}.$$
For  results on perfect Roman domination in
graphs we cite \cite{MR4020540,Perfect-Roman-complexity-2019,Perfect-Roman-Regular-2018,Note-Perfect-Roman-2020}.

\medskip
  A PRDF  of weight $\gamma_R^p(G)$ is called a $\gamma_R^p(G)$-function. Observe that
$\gamma_R(G)\le \gamma_R^p(G)\le 2\gamma^p(G)$ for every graph $G$. Those graphs attaining the equality
$\gamma_R^p(G)=2\gamma^p(G)$ are called \emph{perfect Roman graphs}. All perfect Roman trees were characterized in \cite{PerfectRomanTrees}.

Figure \ref{FigWeakRoman} shows three copies of a graph $G$ with $\gamma_R(G)=\gamma_R^p(G)=4$.   Notice that the labellings  correspond  to the positive weights of all $\gamma_R(G)$-functions.
In particular, the labellings on the center and on the right correspond  to the positive weights of $\gamma_R^p(G)$-functions.

\begin{figure}[!h]
\centering
\begin{tikzpicture}[scale=.6, transform shape]

\node [draw, shape=circle, fill=black] (a1) at  (0,0) {};
\node at (-0.5,0) {\Large $2$};
\node [draw, shape=circle] (a2) at  (1.5,0) {};
\node [draw, shape=circle, fill=black] (a3) at  (3,0) {};
\node at (3,0.6) {\Large $2$};
\node [draw, shape=circle] (a4) at  (4.5,0) {};

\node [draw, shape=circle] (a11) at  (0,1.5) {};
\node [draw, shape=circle] (a12) at  (0,-1.5) {};

\draw(a1)--(a2)--(a3)--(a4);
\draw(a11)--(a1)--(a12);

\node [draw, shape=circle, fill=black] (b1) at  (8,0) {};
\node at (7.5,0) {\Large $2$};
\node [draw, shape=circle] (b2) at  (9.5,0) {};
\node [draw, shape=circle, fill=gray] (b3) at  (11,0) {};
\node at (11,0.6) {\Large $1$};
\node [draw, shape=circle, fill=gray] (b4) at  (12.5,0) {};
\node at (12.5,0.6) {\Large $1$};
\node [draw, shape=circle] (b11) at  (8,1.5) {};
\node [draw, shape=circle] (b12) at  (8,-1.5) {};

\draw(b1)--(b2)--(b3)--(b4);
\draw(b11)--(b1)--(b12);

\node [draw, shape=circle, fill=black] (c1) at  (16,0) {};
\node at (15.5,0) {\Large $2$};
\node [draw, shape=circle] (c2) at  (17.5,0) {};
\node [draw, shape=circle] (c3) at  (19,0) {};
\node [draw, shape=circle, fill=black] (c4) at  (20.5,0) {};
\node at (20.5,0.6) {\Large $2$};
\node [draw, shape=circle] (c11) at  (16,1.5) {};
\node [draw, shape=circle] (c12) at  (16,-1.5) {};

\draw(c1)--(c2)--(c3)--(c4);
\draw(c11)--(c1)--(c12);
\end{tikzpicture}\vspace*{-1mm}
\caption{The labellings associated to the positive weights of all $\gamma_R(G)$-functions on the same graph. The labellings on the center and on the right correspond  to the case of $\gamma_R^p(G)$-functions.}
\label{FigWeakRoman}
\end{figure}
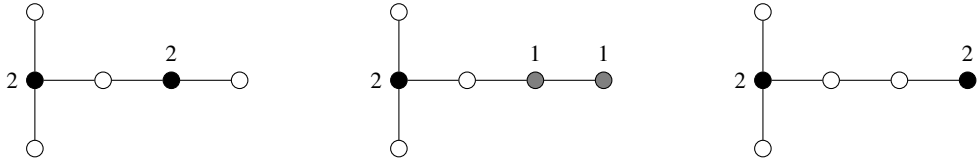

Figure \ref{Fig-Example-PerfectDomination} shows a Roman graph $G$, namely, $\gamma_R(G)=6=2\gamma(G)$. In this case, $\gamma^p(G)=6$ and $\gamma^p_R(G)=9$. The set of labelled vertices form a $\gamma^p(G)$-set and the labels describe the positive weights of a $\gamma_R^p(G)$-function.

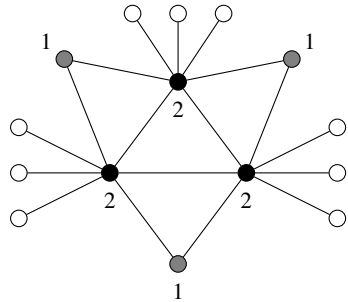
\begin{figure}[ht]
\centering
\begin{tikzpicture}[scale=.6, transform shape]

\node [draw, shape=circle, fill=black] (a1) at  (-1.5,0) {};
\node at (-1.5,-0.6) {\Large $2$};
\node [draw, shape=circle, fill=black] (a2) at  (1.5,0) {};
\node at (1.5,-0.6) {\Large $2$};
\node [draw, shape=circle, fill=black] (a3) at  (0,2) {};
\node at (0,1.4) {\Large $2$};

\node [draw, shape=circle, fill=gray] (a12) at  (0,-2) {};
\node at (0,-2.6) {\Large $1$};

\node [draw, shape=circle, fill=gray] (a13) at  (-2.5,2.5) {};
\node at (-2.9,2.9) {\Large $1$};
\node [draw, shape=circle, fill=gray] (a23) at  (2.5,2.5) {};
\node at (2.9,2.9) {\Large $1$};

\node [draw, shape=circle] (a1h1) at  (-3.5,0) {};
\node [draw, shape=circle] (a1h2) at  (-3.5,1) {};
\node [draw, shape=circle] (a1h3) at  (-3.5,-1) {};

\node [draw, shape=circle] (a2h1) at  (3.5,0) {};
\node [draw, shape=circle] (a2h2) at  (3.5,1) {};
\node [draw, shape=circle] (a2h3) at  (3.5,-1) {};

\node [draw, shape=circle] (a3h1) at  (0,3.5) {};
\node [draw, shape=circle] (a3h2) at  (-1,3.5) {};
\node [draw, shape=circle] (a3h3) at  (1,3.5) {};

\draw(a1)--(a2)--(a3)--(a1);
\draw(a1)--(a12)--(a2)--(a23)--(a3)--(a13)--(a1);

\draw(a1)--(a1h1);
\draw(a1)--(a1h2);
\draw(a1)--(a1h3);

\draw(a2)--(a2h1);
\draw(a2)--(a2h2);
\draw(a2)--(a2h3);

\draw(a3)--(a3h1);
\draw(a3)--(a3h2);
\draw(a3)--(a3h3);

\end{tikzpicture}\vspace*{-1mm}
\caption{The set of labelled vertices form a $\gamma^p(G)$-set and the labels correspond to the positive weights of a $\gamma_R^p(G)$-function.}
\label{Fig-Example-PerfectDomination}\vspace*{-4mm}
\end{figure}

\medskip
The aim of this paper is to obtain closed formulas for the perfect domination number, the Roman domination  number  and the perfect Roman domination number of lexicographic product graphs. The paper is organised as follows. In Section \ref{section-tools} we declare the general notation, terminology and basic tools needed to develop the remaining sections. In Section \ref{Perfect domination and Roman Domination} we obtain closed formulas for the perfect domination number and the Roman domination number of lexicographic product graphs. Finally, Section \ref{Section-Perfect Roman Domination} is devoted to provide tight bounds and closed formulas for the perfect Roman domination number of lexicographic product graphs.

\section{Notation, terminology and basic tools}\label{section-tools}

Throughout the paper, we will use the notation $K_k$  and $N_k$  for  a complete graph and an empty graph   of order $k$, respectively. We use the notation $u \sim v$ if $u$ and $v$ are adjacent vertices,  and $G \cong H$ if $G$ and $H$ are isomorphic graphs. For a vertex $v$ of a graph $G$, $N(v)$ will denote the set of neighbours or \emph{open neighbourhood} of $v$, \emph{i.e.}, $N(v)=\{u \in V(G):\; u \sim v\}$. The \emph{closed neighbourhood}, denoted by $N[v]$, equals $N(v) \cup \{v\}$.  Given a set $S\subseteq V(G)$ and a vertex $v\in S$, the \emph{external private neighbourhood} $\epn(v,S)$ of $v$ with respect to $S$ is defined to be $\epn(v,S)=\{u\in V(G)\setminus S: \, N(u)\cap S=\{v\}\}$.

\medskip
We denote by $\deg(v)=|N(v)|$ the degree of vertex $v$, as well as $\delta(G)=\min_{v \in V(G)}\{\deg(v)\}$ the minimum degree of $G$,   $\Delta(G)=\max_{v \in V(G)}\{\deg(v)\}$ the maximum degree of $G$ and $\n(G)=|V(G)|$ the order of $G$. Given a set $S\subseteq V(G)$, $N(S)=\cup_{v\in S}N(v)$, $N[S]=N(S)\cup S$ and the subgraph of $G$ induced by $S$ will be denoted by $G[S]$.

A set $S\subseteq V(G)$ is a \emph{total dominating set} of a graph $G$ without isolated vertices  if every vertex $v\in V(G)$ is adjacent to at least one vertex in $S$. Let $\mathcal{D}_t(G)$ be the set of total dominating sets of $G$.
The \emph{total domination number} of $G$ is defined to be,
$$\gamma_t(G)=\min\{|S|:\, S\in \mathcal{D}_t(G)\}.$$

 By definition, $\mathcal{D}_t(G)\subseteq \mathcal{D}(G)$, so that $\gamma(G)\le \gamma_t(G)$. Furthermore, $\gamma_t(G)\le 2\gamma(G)$.
 We define a $\gamma_t(G)$-set as a set $S\in \mathcal{D}_t(G)$
with $|S|=\gamma_t(G)$.  The same agreement will be assumed for optimal parameters associated to other characteristic sets defined in the paper.
 For instance, a $\gamma(G)$-set will be a set $S\in \mathcal{D}(G)$
with $|S|=\gamma(G)$.

\medskip
A graph invariant closely related to the domination number is the  packing number. A  set $S\subseteq V(G)$   is a \emph{packing} if $N[u]\cap N[v]= \varnothing$ for every pair of different vertices $u,v\in S$. We define
 $$\wp(G)=\{S\subseteq V(G): S \text{ is a packing  of } G \}.$$
The \emph{packing number}, denoted by  $\rho(G)$, is the maximum cardinality among all packings of $G$, i.e.,
$$\rho(G)=\max \{|S|:\,  S\in \wp(G) \}.$$
Obviously, $\gamma(G)\ge \rho(G)$. Furthermore,  Meir and Moon \cite{MR0401519} showed in 1975 that $\gamma(T)= \rho(T)$ for every tree $T$.  We would point out  that, in general, these $\gamma(T)$-sets and $\rho(T)$-sets  are not identical. Notice that $\mathcal{D}(G) \cap \wp(G)\ne \varnothing$
if and only if there exists a $\gamma(G)$-set which is a $\rho(G)$-set. A graph $G$ is an efficient closed domination graph if $\mathcal{D}(G) \cap \wp(G)\ne \varnothing$.

\medskip
A set $S\subseteq V(G)$ is an \emph{open packing}, if $N(u)\cap N(v)=\varnothing$ for every pair of different  vertices $u,v\in S$.  We define
 $$\wp_o(G)=\{S\subseteq V(G): S \text{ is an open packing  of } G \}.$$
 The \emph{open packing number} of $G$, denoted by $\rho_o(G)$, is the maximum cardinality among all open packings of $G$, i.e.,
$$\rho_o(G)=\max \{|S|:\,  S\in \wp_o(G) \}.$$
By definition, $\wp(G)\subseteq \wp_o(G)$, so that $\rho(G)\le \rho_o(G)$ for every graph $G$, and $\rho_o(G)\le \gamma_t(G)$ for every graph $G$ without isolated vertices.  Besides, if $S\in \wp_o(G)$, then  every vertex of $G[S]$ has degree at most one, which implies that we can write $S=S_0\cup S_1$, where $S_0$ is the set of isolated vertices of $G[S]$ and $S_1=S\setminus S_0$. Obviously, $S_1=\varnothing$ if and only if $S\in \wp(G)$.

\medskip
A graph $G$ is an \emph{efficient open domination graph} if there exists
a set $D$, called an \emph{efficient open dominating set}, for which $V(G)=\cup_{u\in D}N(u)$ and $N(u)\cap N(v)=\varnothing$
for every pair  of distinct vertices $u,v\in D$. As shown in \cite{Kuziak2014e}, if $G$ is an efficient open domination graph with an efficient open dominating set $D$, then $\gamma_t(G)=|D|$.
Hence, the following remark holds.

\begin{remark}\label{efficient open domination graph}
A   graph $G$ is an efficient open domination graph
if and only if  there exists  $S\in \mathcal{D}^p(G)$  such that $G[S] \cong \cup K_2$. In such a case, $|S|=\gamma_t(G)=\rho_o(G)$.
\end{remark}

\begin{corollary}
If $G$ is an efficient open domination graph, then $\gamma^p(G)\le \gamma_t(G)$.
\end{corollary}

Given two nontrivial graphs $G$ and $H$, we define the following properties, which will become  important tools in the next sections.

\begin{enumerate}[]
\itemsep=0.85pt
\item $\mathcal{P}_1(G,H)$: \, $\delta(H)=0$ and  $G$ is an efficient open domination graph.
\item $\mathcal{P}_2(G,H)$: \, $\gamma(H)=1$ and  $G$ is an efficient closed domination graph.
\item $\mathcal{P}_3(G,H)$: \, $\delta(H)=0$, $G$ is an efficient open domination graph and $\gamma^p(G)=\gamma_t(G)$.
\end{enumerate}

Let $f: V(G)\longrightarrow \{0,1,2\}$ be a function on $G$ and let $V_i=\{v\in V(G):\, f(v)=i\}$, where $i\in \{0,1,2\}$. We will identify $f$ with the subsets $V_0,V_1,V_2$, and so we will use the unified notation $f(V_0,V_1, V_2)$ for the function and these associated subsets.

An RDF $f(V_0,V_1, V_2)$ on $G$ is a \emph{total Roman dominating function} if $V_1\cup V_2\in \mathcal{D}_t(G)$   \cite{AbdollahzadehAhangarHenningSamodivkinEtAl2016}. The \textit{total Roman domination number}, denoted by   $\gamma_{tR}(G)$, is the minimum weight among all total Roman dominating functions on $G$. By definition, $\gamma_R(G)\le \gamma_{tR}(G)$.

The \emph{lexicographic product} of two graphs $G$ and $H$ is the graph $G \circ H$ whose vertex set is  $V(G \circ H)=V(G)\times V(H)$ and $(u,v)(x,y) \in E(G \circ H)$ if and only if $ux \in E(G)$ or $u=x$ and $vy \in E(H)$.
For simplicity, the neighbourhood of  $(x,y)\in V(G\circ H)$ will be denoted by $N(x,y)$ instead of $N((x,y))$, and for any PRDF $f$ on $G\circ H$ we will write $f(x,y)$ instead of $f((x,y))$.

\medskip
Notice that  for any $u\in V(G)$  the subgraph of $G\circ H$ induced by $\{u\}\times V(H)$ is isomorphic to $H$. We will denote this subgraph by $H_u$.
For any $u \in V(G)$ and any function $f$ on $G\circ H$ we define $$f(H_u)=\sum_{v\in V(H)}f(u,v) \text{ and }f[H_u]=\sum_{x\in N[u]}f(H_x).$$

For basic properties of the lexicographic product of two graphs we suggest the  books  \cite{Hammack2011,Imrich2000}.
A main problem in the study of product of graphs consists of finding exact values or sharp
bounds for specific parameters of the product of two graphs and express
them in terms of invariants of the factor graphs. In particular,   we cite the following works on domination theory of lexicographic product graphs. For instance, the reader is referred to \cite{MR3363260,Nowakowski1996} for  the domination number, \cite{DD-lexicographic} for the double domination number, \cite{SUmenjak:2012:RDL:2263360.2264103} for the Roman domination number, \cite{TRDF-Lexicographic-2020,Dorota2019} for the total Roman domination number, \cite{MR3057019} for the rainbow domination number, \cite{Dettlaff-LemanskaRodrZuazua2017} for the super domination number, \cite{Valveny2017} for the weak Roman domination number, \cite{TPlexicographic-2019} for the total weak Roman domination number and the secure total domination number, \cite{w-domination} for the Italian  domination number and  \cite{MR3200151} for the doubly connected domination number.

\medskip
 For the remainder of the paper, definitions will be introduced whenever a concept is needed.

\section{Perfect domination and Roman Domination in lexicographic product graphs} \label{Perfect domination and Roman Domination}

The next theorem merges two results obtained in \cite{SUmenjak:2012:RDL:2263360.2264103} and \cite{Zhang2011}.

\begin{theorem}[{\rm\cite{SUmenjak:2012:RDL:2263360.2264103}} and {\rm\cite{Zhang2011}}]\label{teo-char-gamma}

For any graph $G$ with no isolated vertex and any nontrivial graph $H$,

$$\gamma(G\circ H)=\left\{\begin{array}{ll}
                                 \gamma(G) & \mbox{if $\gamma(H)=1$,} \\[5pt]
                                 \gamma_t(G) & \mbox{if $\gamma(H)\geq 2$.}
                               \end{array}\right.$$

\end{theorem}

As the following result shows, when computing the perfect domination number of lexicographic product graphs $G\circ H$, where $G$ is connected and $H$ is not trivial,
we have to take into account that the class of graphs $G\circ H$  satisfies a certain trichotomy, as it is divided into three categories, i.e., the class of graphs $G\circ H$ for which $\mathcal{P}_1(G,H)$ holds, the class of graphs $G\circ H$ for which $\mathcal{P}_2(G,H)$ holds, and the class where neither $\mathcal{P}_1(G,H)$ nor $\mathcal{P}_2(G,H)$ holds.

\begin{theorem}\label{teo-char-gammaPerfectLexic}
For any connected graph $G$ and any nontrivial graph $H$,

$$\gamma^p(G\circ H)=\left\{\begin{array}{ll}
                                 \gamma_t(G) & \text{ if }\mathcal{P}_1(G,H) \text{ holds,}\\[4pt]
                                 \gamma(G) & \text{ if }\mathcal{P}_2(G,H) \text{ holds,} \\[4pt]
                                  \n(G)\n(H) & \text{  otherwise.}
                               \end{array}\right.$$
\end{theorem}

\begin{proof}
Let $S$ be a $\gamma^p(G\circ H)$-set and define  $W_0=\{x\in V(G):\, V(H_x)\cap  S=\varnothing\}$ and $W_1=\{x\in V(G):\, |V(H_x)\cap  S|=1\}$.
We differentiate, the following two cases.

\vspace{0,2cm}
\noindent
Case 1.
There exists $x\in V(G)$ such that $|V(H_x)\cap  S|\ge 2$. Since  $\n(H)\ge 2$, we deduce that $N[x]\times V(H)\subseteq S$, which implies that $S=V(G\circ H)$, i.e., $\gamma^p(G\circ H)=|S|=\n(G)\n(H)$.

\vspace{0,2cm}
\noindent
Case 2. $|V(H_x)\cap  S|\le 1$ for every $x\in V(G)$. Obviously,  $W_1\in \mathcal{D}^p(G)$ and, since $V(H_x)\setminus  S\ne \varnothing$ for every $x\in V(G)$, we conclude that  $S\in \wp_o(G\circ H)$. Let $(x,y)\in S$. If $x$ is an isolated vertex of $G[W_1]$, then $y$ is a universal vertex of $H$, while if $x$ has degree one, then $y$ is an isolated vertex of $H$. Therefore, we have the following two complementary subcases.

\vspace{0,2cm}

\noindent
Subcase 2.1. $\mathcal{P}_1(G,H)$ holds, i.e., $y$ is an isolated vertex of $H$,  $W_1\in \mathcal{D}^p(G)$ and $G[W_1]\cong \cup K_2$. In this case,  Remark \ref{efficient open domination graph} leads to $|W_1|=\gamma_t(G)$. Hence,  $\gamma^p(G\circ H)=|S|=|W_1\times \{y\}|=|W_1|=\gamma_t(G)$.

\vspace{0,2cm}
\noindent
Subcase 2.2. $\mathcal{P}_2(G,H)$ holds, i.e.,  $y$ is a universal vertex of $H$, $W_1$ is $\rho(G)$-set and also a $\gamma(G)$-set. In this case,  $\gamma^p(G\circ H)=|S|=|W_1\times \{y\}|=|W_1|=\gamma(G)$.
\end{proof}

The Roman domination number of the lexicographic product of two connected graphs $G$ and $H$ was studied in \cite{SUmenjak:2012:RDL:2263360.2264103}. Obviously, the connectivity of $G\circ H$ only depends on the connectivity of $G$. Since we need to consider the case where $H$ is not necessarily connected,
 we make next  the necessary modifications to adapt the results obtained in \cite{SUmenjak:2012:RDL:2263360.2264103} to the general case in which $H$ is not necessarily connected.

\begin{lemma}\label{lem-principal}
Let $G$ be a graph with no isolated vertex and  $H$ a  nontrivial graph. Let $f(V_0,V_1,V_2)$ be a $\gamma_R(G\circ H)$-function,  $A_f=\{x\in V(G):\, V(H_x)\cap V_2\ne \varnothing\}$ and $B_f=\{x\in V(G)\setminus A_f:\, V(H_x)\cap V_1\ne \varnothing\}$.
If $|V_2|$ is maximum among all $\gamma_R(G\circ H)$-functions, then $A_f\in \mathcal{D}(G)$  and $B_f=\varnothing$.
\end{lemma}

\begin{proof}
Let $f(V_0,V_1,V_2)$ be a $\gamma_R(G\circ H)$-function such that $|V_2|$ is maximum among all $\gamma_R(G\circ H)$-functions.
If $x\in V(G)\setminus (A_f\cup B_f)$, then $V(H_x)\subseteq V_0$, which implies that $N(x)\cap A_f\neq \varnothing$. Hence, $A_f\cup B_f\in \mathcal{D}(G)$.

\medskip
Now, suppose that there exists $u\in B_f$. Observe that $(N(u)\times V(H))\cap V_2= \varnothing$, and so $V(H_u)\subseteq V_1$.  Given $u'\in N(u)$ and $v\in V(H)$, we define a function  $f'(V_0',V_1',V_2')$ on $G\circ H$ by $f'(H_u)=0$, $f'(u',v)=2$ and $f'(x,y)=f(x,y)$ for the remaining vertices. Notice that $f'$ is a RDF on $G\circ H$ with $|V_2'|>|V_2|$ and, since $H$ is a nontrivial graph,  $f(H_u)=|V(H_u)|\ge 2$,  so that $\omega(f')\leq \omega(f)$, which is a contradiction. Therefore, $B_f=\varnothing$  and  $A_f\in \mathcal{D}(G)$.
\end{proof}

The following result is a direct consequence of Lemma \ref{lem-principal}.

\begin{corollary}\label{LowerBoundRomanLex}
For any graph $G$ without isolated vertices and any nontrivial graph $H$,
$$\gamma_R(G\circ H)\geq 2\gamma(G).$$
\end{corollary}

\begin{theorem}\label{lem-upper-bound}{\rm \cite{SUmenjak:2012:RDL:2263360.2264103}}
For any graph $G$ without isolated vertices and any graph $H$,
$$\gamma_R(G\circ H)\leq 2\gamma_t(G).$$
\end{theorem}

Now, we introduce the definition of domination couple given in \cite{SUmenjak:2012:RDL:2263360.2264103}. We say that an ordered couple $(A,B)$ of disjoint sets $A,B\subseteq V(G)$ is a \emph{dominating couple} of $G$ if every vertex $x\in V(G)\setminus B$ satisfies that  $N(x)\cap (A\cup B)\neq \varnothing$. Also, we define the parameter $\zeta(G)$ as follows.
$$\zeta(G)=\min\{2|A|+3|B|:\, (A,B) \text{ is a dominating couple of } G\}.$$
We say that a dominating couple $(A,B)$ of $G$ is a $\zeta(G)$-couple if $\zeta(G)=2|A|+3|B|$. With this notation in mind, we state the following result.

\begin{theorem}\label{teo-principal-generalizado}
For any graph $G$ without isolated vertices and any nontrivial graph $H$,
$$\gamma_R(G\circ H)=\left\{\begin{array}{ll}
                    2\gamma(G) & \mbox{if $\Delta(H)=\n(H)-1$,} \\[4pt]
                    \zeta(G) & \mbox{if $\Delta(H)=\n(H)-2$,} \\[4pt]
                    2\gamma_t(G) & \mbox{if $\Delta(H)\leq\n(H)-3$.}
                               \end{array}
\right.$$
\end{theorem}

\begin{proof}
As shown in \cite{SUmenjak:2012:RDL:2263360.2264103}, if
 $\gamma(H)=1$ and $G$ is a connected nontrivial graph, then $\gamma_R(G\circ H)= 2\gamma(G).$ Obviously, the same equality holds if $G$ is not connected.

\medskip
 In order to discuss the remaining cases, let $f(V_0,V_1,V_2)$ be a $\gamma_R(G\circ H)$-function  such that $|V_2|$ is maximum.  By Lemma \ref{lem-principal}, $A_f=\{x\in V(G):\, V(H_x)\cap V_2\ne \varnothing\}$ is a dominating set of $G$ and $B_f=\{x\in V(G)\setminus A_f:\, V(H_x)\cap V_1\ne \varnothing\}=\varnothing$. Let $A_{f}'=\{x\in A_f: N(x)\cap A_f=\varnothing\}$.

Assume $\Delta(H)=\n(H)-2$. Since $(A_f\setminus A_{f}', A_{f}')$ is a dominating couple of $G$, we deduce that $\zeta(G)\leq 2|A_f\setminus A_{f}'|+3|A_{f}'|=\omega(f)=\gamma_R(G\circ H)$.  Now, let $v\in V(H)$ be a vertex of maximum degree and $\{v'\}=V(H)\setminus N[v]$. Since for any   $\zeta(G)$-couple $(A,B)$, the function $g(W_0,W_1,W_2)$, defined by $W_2=(A\cup B)\times \{v\}$ and $W_1=B\times\{v'\}$,  is an RDF on $G\circ H$, we deduce that $\gamma_R(G\circ H)\leq \omega(g)=|W_1|+2|W_2|=2|A|+3|B|=\zeta(G)$. Therefore, $\gamma_R(G\circ H)=\zeta(G)$.

\vspace{.2cm}
Finally, assume  $\Delta(H)\leq\n(H)-3$. By Theorem \ref{lem-upper-bound} we only need to prove that $\gamma_R(G\circ H)\geq 2\gamma_t(G).$
In this case, if $x\in A_{f}'$, then $f(H_x)\geq 4$, while if $x\in A_f\setminus A_{f}'$, then  $f(H_x)\geq 2$.
Since $G$ does not have isolated vertices and $A_f\in \mathcal{D}(G)$, we have that  $\gamma_t(G)\le |A_f\setminus A_{f}'|+2| A_{f}'|$. Hence,
$2\gamma_t(G)\leq 2|A_f\setminus A_{f}'|+4|A_{f}'|\leq \omega(f)=\gamma_R(G\circ H)$, which completes the proof.
\end{proof}

 Two simple characterizations of Roman graphs were given in
\cite{Cockayne2004}, but the authors suggest finding classes of Roman graphs. The following result is an immediate consequence of Theorems \ref{teo-char-gamma} and \ref{teo-principal-generalizado}.

\begin{theorem}\label{teo-Roman-graph-case-neq-n(H)-2}
Let $G$ be a graph with no isolated vertex. If $H$ is a graph such that $\Delta(H)\ne \n(H)-2$, then $G\circ H$ is a Roman graph.
\end{theorem}

As $\zeta(G)$ has not been extensively studied, we next obtain tight bounds  on $\gamma_R(G\circ H)$ for  the case in which $\Delta(H)=\n(H)-2$.

\begin{theorem}\label{teo-bound-case=3}
Let $G$ a graph with no isolated vertex and $H$ a graph. If  $\Delta(H)=\n(H)-2$, then  $$\max\{\gamma_{tR}(G),\gamma_t(G)+\gamma(G)\}\leq \gamma_R(G\circ H) \leq \min\{3\gamma(G),2\gamma_t(G)\}.$$
\end{theorem}

\begin{proof}
Let  $f(V_0,V_1,V_2)$ be a $\gamma_R(G\circ H)$-function with $|V_2|$ maximum.  As above, let $A_f=\{x\in V(G):\, V(H_x)\cap V_2\ne \varnothing\}$, $B_f=\{x\in V(G)\setminus A_f:\, V(H_x)\cap V_1\ne \varnothing\}$ and  $A_{f}'=\{x\in A_f: N(x)\cap A_f=\varnothing\}$.
By Lemma \ref{lem-principal}, $B_f=\varnothing$  and  $A_f\in \mathcal{D}(G)$. Furthermore,  if $x\in A_{f}'$, then $f(H_x)=3$, while if $x\in A_f\setminus A_{f}'$, then $f(H_x)=2$. Thus, $$\gamma_R(G\circ H)=3|A_{f}'|+2|A_f\setminus A_{f}'|.$$
We first prove the lower bounds. Let $S\subseteq V(G)$ be a set of minimum cardinality among the sets satisfying  that $A_f\subseteq S$ and $S\cap N(x) \neq \varnothing $  for every vertex $x\in A_{f}'$. Since  $S\in \mathcal{D}_t(G)$, we deduce that  $\gamma_t(G)\le |S|\le 2|A_{f}'|+|A_f\setminus A_{f}'|$. Hence, $\gamma_t(G)+\gamma(G)\leq (2|A_{f}'|+|A_f\setminus A_{f}'|)+ |A_f|= 3|A_{f}'|+2|A_f\setminus A_{f}'|=\gamma_R(G\circ H)$.

\medskip
Now, let $g(W_0,W_1,W_2)$ be a function on $G$ defined by $W_2=A_f$ and $W_1=S\setminus A_f$. Notice that $g$ is a TRDF on $G$. Thus, $\gamma_{tR}(G)\leq \omega(g)=2|A_f|+|S\setminus A_f|\le 3|A_{f}'|+2|A_f\setminus A_{f}'|=\gamma_R(G\circ H)$, which completes the proof of the lower bounds.

In order to prove the upper bounds, let $D$ be a $\gamma(G)$-set, and let $v,v'\in V(H)$ such that $v$ is a  vertex of maximum degree  and   $\{v'\}=V(H)\setminus N[v]$. Notice that the function $f'(V_0',V_1',V_2')$, defined by $V_2'=D\times \{v\}$ and $V_1'=D\times \{v'\}$, is an RDF on $G\circ H$. Therefore, $\gamma_R(G\circ H)\leq \omega(f')=3|D|=3\gamma(G)$.

Finally, the bound $\gamma_R(G\circ H)\le 2\gamma_t(G)$ is already known from Theorem \ref{lem-upper-bound}.   Therefore, the proof is complete.
\end{proof}

The bounds above are tight. Notice that, if $\gamma_t(G)=\gamma(G)$, then $\gamma_R(G\circ H)=\gamma_{tR}(G)=2\gamma_t(G),$ while if $\gamma_t(G)=2\gamma(G)$, then we have $\gamma_R(G\circ H)=\gamma_{t}(G)+\gamma(G)=3\gamma(G).$

\section{Perfect Roman domination in lexicographic product graphs} \label{Section-Perfect Roman Domination}

This section is organised as follows. First we obtain tight bounds on $\gamma_R^p(G\circ H)$  and then we give sufficient and/or necessary conditions for the bounds to be achieved.
 We also discuss the case of perfect Roman graphs and we characterize the graphs where $\gamma_R^p(G\circ H)=\gamma_R(G\circ H)$.

\begin{theorem}
For any graph $G$ without isolated vertices and any graph $H$,
$$\gamma_R^p(G\circ H)\le \gamma^p(G)(\n(H)+1).$$
\end{theorem}

\begin{proof}
Let $S$ be a $\gamma^p(G)$-set and $v\in V(H)$. Let  $f(V_0,V_1,  V_2)$ be a function on $G\circ H$ defined by $V_2=S\times \{v\} $ and $V_1=S\times (V(H)\setminus \{v\})$. Clearly,  $f$ is a PRDF, which implies that $\gamma_R^p(G\circ H)\le \omega (f) =2|S|+|S|(\n(H)-1)=\gamma^p(G)(\n(H)+1)$. Therefore, the result follows.
\end{proof}

In order to see that the bound above is tight, we can consider the corona graph $G\cong G'\odot N_k$, where $k\ge 2$, $G'$ is any graph of minimum degree at least two, and $H$ is a nontrivial graph. In this case, $\gamma_R^p(G\circ H)= \n(G')(\n(H)+1)=\gamma^p(G)(\n(H)+1).$

\begin{theorem}
Let  $G$ be a graph without isolated vertices and $H$  a graph.
The following statements hold.
\begin{enumerate}[{\rm (i)}]
\itemsep=0.9pt
\item For any $\gamma_R^p(G)$-function $f(V_0,V_1,  V_2)$,\vspace{-1mm}
$$\gamma_R^p(G\circ H)\le \gamma^p_R(G)+(|V_1|+|V_2|)(\n(H)-1).$$
\item If there exists a $\gamma_R^p(G)$-function $f(V_0,V_1,  V_2)$ such that $V_2$ is a $\gamma(G)$-set, then\vspace{-1mm}
$$\gamma_R^p(G\circ H)\le \gamma^p_R(G)\n(H)-\gamma(G)(\n(H)-1).$$
\item If  $S$ is a $\gamma^p(G)$-set, $S'=\{x\in S:\, \epn(x,S)=\varnothing\}$ and $S''=S\setminus S'$, then\vspace{-1mm}
$$\gamma_R^p(G\circ H)\le |S'|+2|S''|+ \gamma^p(G)(\n(H)-1).$$
\item If there exists a $\gamma_R^p(G)$-function $f(V_0,V_1,  V_2)$ such that $V_1\cup V_2$ is a $\gamma^p(G)$-set, then\vspace{-1mm}
$$\gamma_R^p(G\circ H)\le \gamma^p_R(G)+\gamma^p(G)(\n(H)-1).$$
\end{enumerate}
\end{theorem}

\begin{proof}
From any $\gamma_R^p(G)$-function $f(V_0,V_1,  V_2)$, we can define a function  $g(W_0,W_1,  W_2)$  on $G\circ H$ as $W_2=V_2\times \{v\} $ and $W_1=V_2\times (V(H)\setminus \{v\})\cup V_1\times V(H)$. It is readily seen that $g$ is a PRDF and, as a result,  $\gamma_R^p(G\circ H)\le \omega (g) =2|V_2|+|V_2|(\n(H)-1)+|V_1|\n(H)=\gamma^p_R(G)+(|V_1|+|V_2|)(\n(H)-1)$. Therefore, (i) follows.

Now, since $\gamma^p_R(G)+(|V_1|+|V_2|)(\n(H)-1)=\gamma^p_R(G)\n(H)-|V_2|(\n(H)-1),$ from (i) we deduce  (ii).

\medskip
In order to prove (iii), we only need to observe that
for any $\gamma^p(G)$-set $S$,
the function $h(V(G)\setminus S, S', S'')$ is a PRDF on $G$. Thus, we conclude the proof of (iii) by analogy to the proof of (i), by using $h$ instead of $f$.

\medskip
Finally, (iv) follows from (i).
\end{proof}

The bounds above are tight. For instance, let $G$ be the graph shown  in Figure \ref{Fig-Example-PerfectDomination},
$V_2=S''$ the set of vertices labelled with $2$,   $V_1=S'$ the set of vertices labelled with $1$ and $V_0=V(G)\setminus (V_1\cup V_2)$. In this case, $V_2$ is a $\gamma(G)$-set, $f(V_0,V_1,V_2)$ is a $\gamma^p_R(G)$-function, $S=S'\cup S''$ is a $\gamma^p(G)$-set and $\gamma^p_R(G\circ H)=6\n(H)+3$ for every  graph $H$.
Therefore, the bounds above are achieved.

\begin{theorem}\label{teo-UpperBoundGeneral}
For any graph $G$ without isolated vertices and any graph $H$,
$$\gamma_R^p(G\circ H)\le \displaystyle\min_{S\in \wp_o(G)}\{ |S_0|(\n(H)-\Delta(H)+1) + |S_1|(2+\delta(H))+\n(H)(\n(G)-|N[S]|)\}.$$
\end{theorem}

\begin{proof}
Let $S=S_0\cup S_1\in \wp_o(G)$ and $y_1, y_2\in V(H)$ such that $\deg(y_1)=\delta(H)$ and $\deg(y_2)=\Delta(H)$. From $S$, $y_1$ and $y_2$, we can construct a function $f(V_0,V_1,V_2)$ on $G\circ H$ as follows. Let $V_2=S_0\times \{y_2\}\cup S_1\times \{y_1\}$ and $V_1=S_0\times (V(H)\setminus N[y_2])\cup S_1\times N(y_1)\cup (V(G)\setminus N[S]) \times V(H)$. It is readily seen that $f$ is a PRDF on $G\circ H$. Therefore, $\gamma_R^p(G\circ H)\le \omega(f)= |S_0|(\n(H)-\Delta(H)+1) + |S_1|(2+\delta(H))+\n(H)(\n(G)-|N[S]|).$ Since the inequality holds for any open packing of $G$, the result follows.
\end{proof}

The following result is an immediate consequence of Theorem \ref{teo-UpperBoundGeneral}.

\begin{corollary}\label{teo-consequences}
Given a graph $G$ without isolated vertices, the following statements hold.
\begin{enumerate}[{\rm (i)}]
\itemsep=0.9pt
\item If $G$ is an efficient open domination graph, then for any graph $H$,\vspace{-1mm}
$$\gamma_R^p(G\circ H)\le \gamma_t(G)(2+\delta(H)).$$
\item If  $G$ is an efficient closed domination graph, then for any graph $H$,\vspace{-1mm}
$$\gamma_R^p(G\circ H)\le  \gamma(G)(\n(H)-\Delta(H)+1).$$
\end{enumerate}
\end{corollary}

\begin{proof}
First, we proceed to prove (i). Let $S\in \mathcal{D}^p(G)$ such that $G[S] \cong \cup K_2$. Notice that $S =S_1\in \wp_o(G)$ and $N[S]=V(G)$. Hence, by Theorem \ref{teo-UpperBoundGeneral} and Remark \ref{efficient open domination graph} we deduce that  $\gamma_R^p(G\circ H)\le |S|(2+\delta(H))=\gamma_t(G)(2+\delta(H)).$

\medskip
Finally, we proceed to prove (ii). Let $S$ be a $\gamma(G)$-set which is a $\rho (G)$-set. Since $S=S_0\in \wp_o(G)$ and $N[S]=V(G)$, by Theorem \ref{teo-UpperBoundGeneral} we deduce that $\gamma_R^p(G\circ H)\le |S|(\n(H)-\Delta(H)+1)= \rho(G)(\n(H)-\Delta(H)+1).$
\end{proof}

As we will show in Theorems  \ref{teo-gamma=1} and \ref{Foirst-Consequence-Bounds}, the bounds above are tight.

\begin{theorem}\label{teo-gamma=1}
Given a  nontrivial graph  $G$ with $\gamma(G)=1$, the following statements hold.
\begin{enumerate}[{\rm (i)}]
\itemsep=0.9pt
\item If $\delta(G)\ge 2$ , then  for any graph $H$,\vspace{-1mm}
$$\gamma_R^p(G\circ H)= \n(H)-\Delta(H)+1.$$
\item If $\delta(G)=1$, then  for any graph $H$,\vspace{-1mm}
$$\gamma_R^p(G\circ H)=\min\{2\delta(H)+4,\n(H)-\Delta(H)+1\}.$$
\end{enumerate}
\end{theorem}

\begin{proof}
Let $f(V_0,V_1,V_2)$ be a $\gamma_R^p(G\circ H)$-function.
We assume first that $\delta(G)\ge 2$. Notice that, in such a case, $N(x)\cap N(x')\ne \varnothing$ for any $x,x'\in V(G)$.  We differentiate three  cases for $V_2$.

\eject
\noindent
Case 1. There exists $x\in V(G)$ such that  $|V_2 \cap V(H_x)|\ge 2$. In this case, $f(H_{x'})=\n(H)$ for every $x'\in N(x)$, and so   $\gamma_R^p(G\circ H)=\omega(f)\ge f[H_x]\ge 4+\n(H)$, which is a contradiction with Corollary \ref{teo-consequences}-(ii).

\vspace{0,2cm}
\noindent Case 2: There exist two different vertices $(x,y),(x',y')\in V_2$ such that $x\neq x'$. In this case, $f(H_{z})=\n(H)$ for any $z\in N(x)\cap N(x')$, and so   $\gamma_R^p(G\circ H)=\omega(f)\ge f[H_z]\ge 4+\n(H)$, which is again a contradiction with Corollary \ref{teo-consequences}-(ii).

\vspace{0,2cm}

\noindent Case 3: $V_2=\{(x,y)\}$. In this case,  $f(x,v)=1$ for every $v\in V(H)\setminus N[y]$. Hence, $\gamma_R^p(G\circ H)=\omega(f)\ge f(H_x)\ge \n(H)-\deg(y)+1\ge \n(H)-\Delta(H)+1$. By Corollary \ref{teo-consequences}-(ii) we conclude that $\gamma_R^p(G\circ H)=\n(H)-\Delta(H)+1.$

According to the three cases above, (i) follows.

From now on we assume that $\delta(G)=1$ and we consider the following  three   cases for $V_2$.

\vspace{0,2cm}
\noindent Case 1': There exists $x\in V(G)$ such that  $|V_2 \cap V(H_x)|\ge 2$. As in Case 1, we obtain a contradiction.

\vspace{0,2cm}
\noindent Case 2': There exist two different vertices $(x,y),(x',y')\in V_2$ such that $x\ne x'$.  If   $\deg(x)<\Delta(G)-1$ and $\deg(x')<\Delta(G)-1$, then $f(H_{z})=\n(H)$ for every $z\in N(x)\cap N(x')$, and so   $\gamma_R^p(G\circ H)=\omega(f)\ge f[H_z]\ge 4+\n(H)$, which is a contradiction with Corollary \ref{teo-consequences}-(ii).

\medskip
Now, assume that  $\deg(x)=\Delta(G)-1$. If $\deg(x')\ge 2$, then as above $f(H_{z})=\n(H)$ for every $z\in N(x)\cap N(x')$, and we have again a contradiction with Corollary \ref{teo-consequences}-(ii). Finally, if $\deg(x')=1$, then $f(x,b)\ge 1$ for every $b\in N(y)$ and $f(x',b')\ge 1$ for every $b'\in N(y')$. Thus,
 $\gamma_R^p(G\circ H)=\omega(f)\ge f(H_x)+f(H_{x'})\geq 2\delta(H)+4$, and by Corollary  \ref{teo-consequences}-(i) we conclude that  $\gamma_R^p(G\circ H)= 2\delta(H)+4$.

\vspace{0,2cm}
\noindent Case 3': $V_2=\{(x,y)\}$. As in Case 3, we deduce that $\gamma_R^p(G\circ H)=\n(H)-\Delta(H)+1.$

\medskip
According to these last three cases, (ii) follows.
\end{proof}

\begin{lemma}\label{Todos0o1}
Let $f(V_0,V_1,V_2)$ be a $\gamma_R^p(G\circ H)$-function and $x\in V(G)$. If $V(H_x)\cap V_2=\varnothing$, then either $V(H_x)\subseteq V_0$ or $V(H_x)\subseteq V_1$.
\end{lemma}

\begin{proof}
Suppose that $V(H_x)\cap V_2=\varnothing$ and there exist $y_1,y_2\in V(H)$ such that $f(x,y_1)=0$ and $f(x,y_2)=1$. In such a case, there exists exactly one vertex $(u,v)\in V_2$ which is adjacent to $(x,y_1)$. Hence, $u\in N(x)$ and  $(u,v)$ is the only vertex belonging to $V_2$ which is adjacent to $(x,y_2)$. Thus, the function $g(W_0,W_1,W_2)$, defined by $W_2=V_2$, $W_1=V_1\setminus V(H_x)$ and $W_0=V_0\cup V(H_x)$, is a PRDF on $G\circ H$, which is a contradiction, as $\omega(g)<\omega(f)$. Therefore, the result follows.
\end{proof}

\begin{theorem}\label{teo-lowerBoundGeneral}
For any graph $G$ without isolated vertices and any nontrivial graph $H$,
$$\gamma_R^p(G\circ H)\ge \gamma(G)\min\{\n(H)-\Delta(H)+1,2+\delta(H)\}.$$
\end{theorem}

\begin{proof}
Let $f(V_0,V_1,V_2)$ be a $\gamma_R^p(G\circ H)$-function, and define $W_0=\{x\in V(G): V(H_x)\subseteq V_0\}$, $W_1=\{x\in V(G): V(H_x)\subseteq V_1\}$ and $W_2=V(G)\setminus (W_0\cup W_1)$.
In fact, by Lemma \ref{Todos0o1}, $W_2=\{x\in V(G):\, V(H_x)\cap V_2\ne \emptyset \}$.
Let $W_{2,0}$ be the set of isolated vertices of $G[W_2]$, $W_{2,1}=W_2\setminus W_{2,0}$ and $W_{2,0}^0=\{x\in W_{2,0}: N(x)\times V(H)\cap V_0\neq \varnothing \}$.

Thus, if $x\in W_0$, then $V(H_x)\subseteq V_0$ and there exists exactly one vertex $(u,v)\in V_2$ such that $u\in N(x)\cap W_2$. Also, if $x\in W_{2,0}\setminus W_{2,0}^0$, then $N(x)\cap W_1\neq \varnothing $.
Hence, $W_1\cup W_{2,1}\cup W_{2,0}^0\in \mathcal{D}(G)$. Notice that if $x\in W_{2,0}^0$, then $f(H_x)\geq \n(H)-\Delta(H)+1$ and if $x\in W_{2,1}$, then $f(H_x)\geq 2+\delta(H)$. Therefore,
 \begin{eqnarray*}
\gamma_R^p(G\circ H)&=&\sum_{x\in V(G)}f(H_x)\\
&\ge &\sum_{x\in W_{2,0}^0}f(H_x)+\sum_{x\in W_{2,1}}f(H_x)+\sum_{x\in W_1}f(H_x)\\
&\geq & |W_{2,0}^0|(\n(H)-\Delta(H)+1)+|W_{2,1}|(2+\delta(H))+|W_1|\n(H)\\
&\geq &(|W_{2,0}^0|+|W_{2,1}|+|W_1|)\min\{\n(H)-\Delta(H)+1,2+\delta(H)\}\\
&\geq &\gamma(G)\min\{\n(H)-\Delta(H)+1,2+\delta(H)\}.
\end{eqnarray*}

\vspace*{-8mm}
\end{proof}

From Corollary \ref{teo-consequences} and Theorem \ref{teo-lowerBoundGeneral} we deduce the following result.

\begin{theorem}\label{Foirst-Consequence-Bounds}
Given a graph $G$  without isolated vertices, the following statements hold.
\begin{enumerate}[{\rm (i)}]
\itemsep=0.9pt
\item If   $G$ is an efficient closed domination graph, then  for any  graph $H$ with $2\le \n(H)\le \Delta(H)+\delta(H)+1$,\vspace{-1mm}
$$\gamma_R^p(G\circ H)= \gamma(G)(\n(H)-\Delta(H)+1).$$
\item If $\gamma^p(G)=\gamma_t(G)=\gamma(G)$  and  $G$ is an efficient open domination graph, then for any nontrivial graph  $H$ with $ \n(H)\ge \Delta(H)+\delta(H)+1$,\vspace{-1mm}
$$\gamma_R^p(G\circ H)= \gamma(G)(2+\delta(H)).$$
\end{enumerate}
\end{theorem}

\begin{corollary}\label{CorollaryFoirst-Consequence-Bounds}
Given a graph $G$  without isolated vertices and a nontrivial graph $H$, the following statements hold.
\begin{enumerate}[{\rm (i)}]
\itemsep=0.9pt
\item If  $\mathcal{P}_2(G,H)$ holds, then $\gamma_R^p(G\circ H)= 2\gamma(G).$
\item If $\gamma^p(G)=\gamma(G)$  and   $\mathcal{P}_3(G,H)$ holds, then
$\gamma_R^p(G\circ H)= 2\gamma(G).$
\end{enumerate}
\end{corollary}

\begin{theorem}\label{TrivialLowerbound}

Given two nontrivial graphs $G$ and $H$, the following statements hold.
\begin{enumerate}[{\rm (i)}]
\itemsep=0.9pt
\item $\gamma_R^p(G\circ H)\ge \max\{\gamma_R^p(G), 2\gamma(G)\}.$
\item $\gamma_R^p(G\circ H)=\gamma_R^p(G)$ if and only if $\gamma_R^p(G)=2\gamma^p(G)$ and either $\mathcal{P}_2(G,H)$ holds or $\mathcal{P}_3(G,H)$ holds.
\item If $H$ has order at least three, then $\gamma_R^p(G\circ H)=2\gamma(G)$ if and only if $\gamma^p(G)=\gamma(G)$ and either $\mathcal{P}_2(G,H)$ holds or $\mathcal{P}_3(G,H)$ holds.
\end{enumerate}
\end{theorem}

\begin{proof}
By Theorem \ref{teo-lowerBoundGeneral} we deduce that $\gamma^p_R(G\circ H)\ge 2\gamma(G).$
From now on, let $f(V_0,V_1,V_2)$ be a $\gamma_R^p(G\circ H)$-function, and define the function  $g(W_0,W_1,W_2)$  on $G$ by $W_0=\{x\in V(G):\, V(H_x)\subseteq V_0\}$,   $W_1=\{x\in V(G):\, V(H_x)\subseteq V_1\}$ and $W_2=V(G)\setminus (W_0\cup W_1)$. If $x\in W_0$, then $V(H_x)\subseteq V_0$ and there exists exactly one vertex $(u,v)\in V_2$ such that $u\in N(x)\cap W_2$. Hence, $g$ is a PRDF on $G$, and so $\gamma_R^p(G)\le \omega(g)\le \omega (f)=\gamma_R^p(G\circ H).$ Therefore, (i) follows.

In order to prove (ii), assume   $\gamma_R^p(G\circ H)=\gamma_R^p(G)$. Notice that in this case the function $g(W_0,W_1,W_2)$ defined above is a $\gamma_R^p(G)$-function. We first  show that the $\gamma_R^p(G\circ H)$-function $f(V_0,V_1,V_2)$ satisfies $V_1=\varnothing$. Suppose to the contrary, that there exists $(u,v)\in V_1$. If $V(H_u)\cap V_2= \varnothing$, then by Lemma \ref{Todos0o1} we have that $V(H_u)\subseteq V_1$ and since $|V(H_u)|\ge 2$, we deduce that $\gamma_R^p(G)\le \omega(g)<\omega (f)=\gamma_R^p(G\circ H),$ which is a contradiction. The same contradiction is reached if $V(H_u)\cap V_2\ne\varnothing$, as in such a case $f(H_u)\ge 3$. Hence, $V_1=\varnothing$, which implies that $W_1=\varnothing$ and $W_2\in \mathcal{D}^p(G)$.

\medskip
Furthermore,  $2 \gamma^p(G)\le 2|W_2|\le \gamma_R^p(G\circ H)=\gamma_R^p(G)\le 2 \gamma^p(G)$, and so we conclude that $W_2$ is a  $\gamma^p(G)$-set and $\gamma_R^p(G)=2 \gamma^p(G)$.
We differentiate two cases for $x\in W_2$.

\vspace{0,2cm}
\noindent
Case 1. There exists $x'\in N(x)\cap W_2$.
In this case, there exist $y,y'\in V(H)$ such that  $(x,y),(x',y')\in V_2$, and so no vertex in  $V(H_{x})\setminus \{(x,y)\}$ is adjacent to $(x,y)$. Hence $y$ is an isolated vertex of $H$. Notice that   $N(x)\cap W_2=\{x'\}$, otherwise every vertex in $V(H_{x})\cap V_0=V(H_x)\setminus \{(x,y)\}$ is adjacent to two vertices in $V_2$, which is a contradiction.

\vspace{0,2cm}
\noindent
Case 2. $N(x)\cap W_2=\varnothing$. In this case, there exists $y\in V(H)$ such that $(x,y)\in V_2$ and every vertex in  $V(H_{x})\cap V_0=V(H_x)\setminus \{(x,y)\}$ has to be adjacent to $(x,y)$. Hence, $y$ is a universal vertex of $H$  and so $\gamma(H)=1$. Notice also that $N(x)\cap N(x')=\varnothing$  for every $x'\in W_2\setminus \{x\}.$

\medskip
According to the two cases above, either $H$ has at least one isolated vertex or $\gamma(H)=1$. Thus, either Case 1 holds for every vertex  $x\in W_2$ or Case 2 holds for every vertex  $x\in W_2$. In the first case, it is readily seen that $\mathcal{P}_3(G,H)$ holds, while if  Case 2 holds for every vertex  $x\in W_2$, then  $W_2$ is a packing, and so $\gamma^p(G)=|W_2|\le \rho (G)\le \gamma(G)\le \gamma^p(G)$, which implies that $\mathcal{P}_2(G,H)$ holds.

Conversely, assume that $\gamma_R^p(G)=2\gamma^p(G)$. If $\mathcal{P}_3(G,H)$ holds, then Corollary \ref{teo-consequences}-(i) and the lower bound (i) lead to $\gamma_R^p(G\circ H)=\gamma_R^p(G)$. Finally, if $\mathcal{P}_2(G,H)$ holds, then  Theorem \ref{Foirst-Consequence-Bounds}-(i) leads to
 $\gamma_R^p(G\circ H)=\gamma^p_R(G)$, which completes the proof of (ii).

\medskip
 We proceed to prove (iii). Assume  $\gamma_R^p(G\circ H)=2\gamma (G)$. Since  $|V(H)|\ge 3$ and $W_1\cup W_2\in \mathcal{D}(G)$, we deduce that if $W_1\ne \varnothing$, then  $2\gamma(G)< |V(H)||W_1|+2|W_2|\le \omega (f)=\gamma_R^p(G\circ H),$ which is a contradiction. Hence,  $W_1=\varnothing$ and $W_2\in \mathcal{D}^p(G)$.
Furthermore,  $2 \gamma(G)\le 2|W_2|=\gamma_R^p(G\circ H)= 2 \gamma(G)$, which  implies  that $W_2$ is a  $\gamma(G)$-set and also a $\gamma^p(G)$-set.
We differentiate two cases for $x\in W_2$.

\vspace{0,2cm}
\noindent
Case 1'. There exists $x'\in N(x)\cap W_2$.
As in Case 1, we can see that  $H$ has an isolated vertex and   $N(x)\cap W_2=\{x'\}$.

\vspace{0,2cm}
\noindent
Case 2'. $N(x)\cap W_2=\varnothing$. By analogy to Case 2  we deduce that $\gamma(H)=1$.

\medskip
Thus, either Case 1' holds for every vertex  $x\in W_2$ or Case 2' holds for every vertex  $x\in W_2$.
In the first case, we deduce that $\mathcal{P}_3(G,H)$ follows, while if  Case 2' holds for every vertex  $x\in W_2$, then  $W_2$ is a packing, and so $\gamma^p(G)=|W_2|\le \rho (G)\le \gamma(G)\le \gamma^p(G)$, which leads to $\mathcal{P}_2(G,H)$.

Conversely, assume $\gamma^p(G)=\gamma(G)$. If $\mathcal{P}_3(G,H)$ holds, then Corollary \ref{teo-consequences}-(i) and the lower bound (i) lead to  $\gamma_R^p(G\circ H)=2\gamma(G)$. Finally, if $\mathcal{P}_2(G,H)$ holds, then Theorem \ref{Foirst-Consequence-Bounds}-(i) leads to
 $\gamma_R^p(G\circ H)=2\gamma(G)$, which completes the proof.
\end{proof}

\begin{theorem}\label{Segmentos-Aislados}
 Let $G$ and $H$ be two graphs.
If $G$ is an efficient open domination graph and $\n(H)\ge \Delta(H)+2\delta(H)+3$, then
$$\gamma_R^p(G\circ H)=\gamma_t(G)(2+\delta(H)).$$
\end{theorem}
		
\begin{proof}
Let $S\in \mathcal{D}^p(G)$  such that $G[S] \cong \cup K_2$ and assume that $\n(H)\ge \Delta(H)+2\delta(H)+3$.

\medskip
Let $x,x'\in S$ be two adjacent vertices, and define  $X_x=\{x\}\cup \epn(x,S)=N[x]\setminus \{x'\}$ and $X_{x'}=\{x'\}\cup \epn(x',S)=N[x']\setminus \{x\}$. Let $f(V_0,V_1,V_2)$ be a $\gamma_R^p(G\circ H)$-function and define $\varepsilon(x,x')=f(X_x\times V(H))+f(X_{x'}\times V(H))$. In order to prove  that $\varepsilon(x,x')\ge 2(2+\delta(H))$, we differentiate the following cases.

\vspace{0,2cm}
\noindent Case 1: $V_2\cap V(H_x)=V_2\cap V(H_{x'})=\varnothing $. By Lemma \ref{Todos0o1} we have that  $V(H_x)\subseteq V_0$ or $V(H_x)\subseteq V_1$, and also
$V(H_{x'})\subseteq V_0$ or $V(H_{x'})\subseteq V_1.$
The case $V(H_x)\subseteq V_1$ and $V(H_{x'})\subseteq V_1$ leads to $\varepsilon(x,x')\ge f(H_x)+f(H_{x'})=2|V(H)|\ge 2(2+\delta(H))$.

\medskip
If $V(H_x)\subseteq V_0$ and $V(H_{x'})\subseteq V_1$, then  $|V_2\cap \left(\epn(x,S)\times V(H)\right)|=1$, which implies that
$\varepsilon(x,x')\ge  f(X_x\times V(H))+|V(H_{x'})|\ge 2+|V(H)|\ge 5+ 2\delta(H)  > 2(2+\delta(H))$.

\medskip
Finally, if  $V(H_x)\subseteq V_0$ and $V(H_{x'})\subseteq V_0,$ then $|V_2\cap \left(\epn(x,S)\times V(H)\right)|=1$ and also\linebreak
 $|V_2\cap \left(\epn(x',S)\times V(H)\right)|=1$. Since  $\n(H)\ge \Delta(H)+2\delta(H)+3$,
the  vertex of weight two in $\epn(x,S)\times V(H)$ is not able to dominate every vertex in $\epn(x,S)\times V(H)$, which implies that  $f(\epn(x,S)\times V(H))\ge \delta(H)+2$ or $f(\epn(x,S)\times V(H))\ge \n(H)-\Delta(H)+1\ge 2(2+\delta(H))$. By applying the same reasoning to $ \epn(x',S)\times V(H)$ we conclude that $\varepsilon(x,x')\ge  f(\epn(x,S)\times V(H))+f(\epn(x',S)\times V(H))\ge  2(2+\delta(H))$.

\medskip 
\noindent Case 2: $V_2\cap V(H_x)\ne\varnothing$ and $ V_2\cap V(H_{x'})=\varnothing $. By Lemma \ref{Todos0o1}, either  $V(H_{x'})\subseteq V_0$ or $V(H_{x'})\subseteq V_1$.
If $V(H_{x'})\subseteq V_1$, then
$\varepsilon(x,x')\ge  2+|V(H_{x'})|\ge 5+ 2\delta(H)  > 2(2+\delta(H))$. Now, assume $V(H_{x'})\subseteq V_0$. In this case,$|V_2\cap V(H_x)|=1$ and, since $\n(H)\ge \Delta(H)+2\delta(H)+3$, we have that $(V(H_x)\setminus N[x])\cap V_0= \emptyset $ or $|V_2\cap (\epn(x,S)\times V(H))|=1$. In both cases we deduce that
$\varepsilon(x,x')\ge    2(2+\delta(H))$.

\medskip 
\noindent Case 3: $V_2\cap V(H_x)\ne\varnothing$ and $ V_2\cap V(H_{x'})\ne \varnothing $. In this case,
$|V_2\cap V(H_x)|=1$ and $|V_2\cap V(H_x)|=1$, which implies that $\varepsilon(x,x')\ge  f(H_x)+f(H_{x'})\ge   2(2+\delta(H))$.

\medskip
According to the three cases above we conclude that $\varepsilon(x,x')\ge 2(2+\delta(H))$ for every pair of adjacent vertices $x,x'\in S$. Hence,
$$\gamma_R^p(G\circ H)=\omega(f)\ge \sum_{x\in S}f(X_x\times V(H))\ge |S|(2+\delta(H))= \gamma_t(G) (2+\delta(H)).$$
Therefore, Corollary \ref{teo-consequences}-(i) leads to $\gamma_R^p(G\circ H)= \gamma_t(G)(2+\delta(H)).$
\end{proof}

From the following inequalities we can derive results on the perfect Roman domination number of $G\circ H$.
$$\gamma_R(G\circ H)\le \gamma_R^p(G\circ H)\le 2\gamma^p(G\circ H).$$
Next we discuss the cases in which  the bounds are sharp.

\begin{theorem}\label{LexicPerfectRomanGraphs}
Given a connected nontrivial graph $G$ and $H$ a nontrivial graph, the following statements hold.
\begin{enumerate}[{\rm (i)}]
\itemsep=0.9pt
\item If $\Delta(H)= \n(H)-1$, then $\gamma_R^p(G\circ H)=2\gamma^p(G\circ H)$ if and only if  $\mathcal{P}_2(G,H)$  holds.
\item If $\Delta(H)= \n(H)-2$, then the following statements hold.

\begin{itemize}
\item[{\rm (a)}] If $\gamma_R^p(G\circ H)=2\gamma^p(G\circ H)$, then $\mathcal{P}_1(G,H)$ holds and $2\gamma_t(G)\le 2|S_1|+3|S_0|$ for every $S\in \wp_o(G)\cap \mathcal{D}(G)$.
\item[{\rm (b)}] If $\gamma_R^p(G)=2\gamma_t(G)$ or $\gamma_t(G)=\gamma(G)$, then $\gamma_R^p(G\circ H)=2\gamma^p(G\circ H)$ if and only if $\mathcal{P}_1(G,H)$ holds.
\end{itemize}

\item If $\Delta(H)\le \n(H)-3$, then $\gamma_R^p(G\circ H)=2\gamma^p(G\circ H)$ if and only if $\mathcal{P}_1(G,H)$ holds.
\end{enumerate}
\end{theorem}

\begin{proof}
Assume $\gamma_R^p(G\circ H)=2\gamma^p(G\circ H)$.
Since $G$ is a graph without isolated vertices and $H$ a nontrivial graph, $\gamma_R^p(G\circ H)<\n(G)\n(H)$, so that from  Theorem \ref{teo-char-gammaPerfectLexic} we have that
either  $\mathcal{P}_1(G,H)$ holds or  $\mathcal{P}_2(G,H)$ holds.  Notice that, by definition, $\mathcal{P}_2(G,H)$ is associated with $\Delta(H)= \n(H)-1$.

\medskip
Now, if   $\mathcal{P}_2(G,H)$ holds, then Theorem \ref{teo-char-gammaPerfectLexic} leads to
 $\gamma_R^p(G\circ H)\le 2\gamma^p(G\circ H)=2\gamma(G)$. In such a case, from Theorem \ref{TrivialLowerbound}-(i)
we conclude that $\gamma_R^p(G\circ H)= 2\gamma^p(G\circ H).$ Therefore, (i) follows.

\medskip
From now on, assume that   $\mathcal{P}_1(G,H)$ holds.  Notice that, in this case,  Theorem \ref{teo-char-gammaPerfectLexic} leads to
\begin{equation}\label{Eq-poof-caso chungo}
\gamma_R^p(G\circ H)\le 2\gamma^p(G\circ H)=2\gamma_t(G).
\end{equation}

First, consider the case $\Delta(H)= \n(H)-2$. Let $v,v'\in V(H)$  such that  $\deg(v)=\n(H)-2$ and  $\deg(v')=0$. Now, if there exists $S\in \wp_o(G)\cap \mathcal{D}(G)$ such that  $2\gamma_t(G)> 2|S_1|+3|S_0|$ then the function
$g(X_0,X_1,X_2)$, defined by $X_2=S_1\times \{v'\}\cup S_0\times \{v\}$ and $X_1=S_0\times \{v'\}$, is a PRDF on $G\circ H$, and so $\gamma_R^p(G\circ H)\leq \omega(g)= 2|S_1|+3|S_0|<2\gamma_t(G)=2\gamma^p(G\circ H)$.
Therefore, (ii)-(a) follows.

\medskip
Furthermore, if
$\gamma_R^p(G)=2\gamma_t(G)$, then by Theorem \ref{TrivialLowerbound}-(i), $2\gamma_t(G)=\gamma_R^p(G)\le \gamma_R^p(G\circ H)$ and so Eq.\ \eqref{Eq-poof-caso chungo}  implies that  $\gamma_R^p(G\circ H)=2\gamma^p(G\circ H).$ The case $\gamma_t(G)=\gamma(G)$ is analogous to the previous one.  Therefore, (ii)-(b) follows.

\medskip
Finally, if $\Delta(H)\le \n(H)-3$, then by Theorem \ref{Segmentos-Aislados} we have that $\gamma_R^p(G\circ H)= 2\gamma_t(G)$. Hence, Eq.\ \eqref{Eq-poof-caso chungo}  implies that $\gamma_R^p(G\circ H)= 2\gamma^p(G\circ H)$, which completes the proof of (iii).
 \end{proof}

In order to state the next result, we define the following parameter.
$$\zeta'(G)= \displaystyle\min_{S\in \wp_o(G)\cap \mathcal{D}(G)}\{ 4|S_0| + 2|S_1|\}.$$
A set $S\in \wp_o(G)\cap \mathcal{D}(G)$ of cardinality $|S|=\zeta'(G)$ will be called a $\zeta'(G)$-set.

\medskip
The following  straightforward lemma will be used in the proof of our next result.

\begin{lemma}\label{teo-trivial-char}
A graph  $G$ is a perfect Roman graph if and only if there exists a $\gamma_R^p(G)$-function $f(V_0,V_1,V_2)$ such that $V_1=\varnothing $.
\end{lemma}

\begin{theorem}
The following statements hold for a connected nontrivial graph $G$ and any   graph $H$  of order at least three.

\begin{enumerate}[{\rm (i)}]
\itemsep=0.95pt
\item If $\Delta(H)=\n(H)-1$, then $\gamma_R^p(G\circ H)=\gamma_R(G\circ H)$ if and only if  $\mathcal{P}_2(G,H)$  holds.
\item If $\Delta(H)=\n(H)-2$, then $\gamma_R^p(G\circ H)=\gamma_R(G\circ H)$ if and only if  there exists a $\zeta(G)$-couple $(A,B)$ such that $A\cup B\in \wp_o(G)$ and $A=\varnothing$ whenever  $\delta(H)\ge 1$.

\item If $\Delta(H)=\n(H)-3$, then $\gamma_R^p(G\circ H)=\gamma_R(G\circ H)$ if and only if either $\delta(H)=0$ and  $\zeta'(G)=2\gamma_t(G)$ or $\delta(H)\ge 1$ and $\gamma_t(G)=2\gamma^p(G)=2\rho(G)$.

\item If $\Delta(H)\leq \n(H)-4$, then $\gamma_R^p(G\circ H)=\gamma_R(G\circ H)$ if and only if  $\mathcal{P}_1(G,H)$  holds.
\end{enumerate}
\end{theorem}

\begin{proof}
First,  assume  $\gamma_R^p(G\circ H)=\gamma_R(G\circ H)$. Let $f(V_0,V_1,V_2)$ be a $\gamma_R^p(G\circ H)$-function, and define $W_0=\{x\in V(G): V(H_x)\subseteq V_0\}$, $W_1=\{x\in V(G): V(H_x)\subseteq V_1\}$ and $W_2=V(G)\setminus (W_0\cup W_1)$. Notice that $f$ is also a $\gamma_R(G\circ H)$-function.
If there exists $u\in W_1$, then for any $u'\in N(u)$ and $v\in V(H)$, the function $g$, defined by $g(H_u)=0$, $g(u',v)=2$, and $g(x,y)=f(x,y)$ for the remaining vertices, is an RDF on $G\circ H$ of weight $\omega(g)<\omega(f)=\gamma_R^p(G\circ H)=\gamma_R(G\circ H)$, which is a contradiction. Hence,  $W_1=\varnothing $.
Now, suppose that $G[W_2]$ has a vertex $x$ of degree at least two. Since $f$  is a  $\gamma_R^p(G\circ H)$-function, $V(H_x)\cap V_0=\varnothing$ and, since $f$  is a  $\gamma_R(G\circ H)$-function, $V(H_x)\cap V_1=\varnothing$, which is a contradiction. Therefore, $W_2\in \wp_o(G)\cap \mathcal{D}^p(G)$. We differentiate two cases.
From each case, we will get partial conclusions and, once both cases have been analysed, we will be able to complete the proof of each statement separately.

\vspace{.2cm}
\noindent
Case 1.  $V_1=\varnothing$. Lemma \ref{teo-trivial-char} leads to $\gamma_R^p(G\circ H)=2\gamma^p(G\circ H)$. Hence, if $\Delta(H)=\n(H)-1$, then by Theorem \ref{LexicPerfectRomanGraphs}-(i) we deduce that $\mathcal{P}_2(G,H)$  holds.  Analogously, if $\Delta(H)\leq \n(H)-2$, then by Theorem \ref{LexicPerfectRomanGraphs}-(ii) we deduce that $\mathcal{P}_1(G,H)$ holds. Notice that in this latter case  $\delta(H)=0$ and also Theorem \ref{teo-char-gammaPerfectLexic} leads to $\gamma_R(G\circ H)=\gamma_R^p(G\circ H)=2\gamma^p(G\circ H)=2\gamma_t(G)$.

\vspace{.2cm}
\noindent
Case 2. $V_1\neq \varnothing$. Let $u\in V(G)$   such that $V(H_u)\cap V_1\neq \varnothing $. Since $W_1=\varnothing$, by Lemma \ref{Todos0o1} we have that $u\in W_2$.
Since $f$ is also a $\gamma_R(G\circ H)$-function,  $N(u)\cap  W_2=\varnothing$ and so $N(u)\subseteq W_0$, which implies that  $W_2'=\{x\in W_2:\, V(H_x)\cap V_1\ne  \varnothing \}$ is a packing. Notice also that $|V(H_u)\cap V_2|=1$.
With these facts in mind, we differentiate the following subcases.

\vspace{.2cm}
\noindent
Subcase 2.1. $\Delta(H)=\n(H)-1$.  Let $V(H_u)\cap V_2=\{(u,v)\}$ and $(u,v')\in  V(H_u)\cap V_1$. Notice, that $v'\not \in N(v)$, as $f$ is a $\gamma_R(G\circ H)$-function. Now, let  $v''$ be a universal vertex of $H$ and define a function $g'$ as $g'(u,v'')=2$, $g'(u,v)=g'(u,v')=0$  and $g'(x,y)=f(x,y)$ for the remaining vertices. Obviously, $g'$ is an RDF  on $G\circ H$ with $\omega(g')<\omega(f)=\gamma_R(G\circ H)$, which is  a contradiction. Therefore, $\Delta(H)=\n(H)-1$ leads to $V_1= \varnothing$.

\vspace{.2cm}
\noindent
Subcase 2.2. $\Delta(H)= \n(H)-2$. By Theorem \ref{teo-principal-generalizado}, $\gamma_R(G\circ H)=\zeta(G)$, and since $W_2\in \wp_o(G)\cap \mathcal{D}^p(G)$, we have that $(W_2\setminus W_2',W_2')$ is a dominating couple, which implies that  $\gamma_R(G\circ H)=\zeta(G)\le 2|W_2\setminus W_2'|+3|W_2'|=2|W_2|+|W_2'| \le 2|V_2|+|V_1|=\gamma_R(G\circ H)$, which implies that $(W_2\setminus W_2',W_2')$ is a $\zeta(G)$-couple.

\medskip
Now, assume  $\delta(H)\ge 1$.
Suppose that there exists $x\in W_2\setminus W_2'$,  and let $(x,y)\in V_2$. In such a case, $(V(H_x)\setminus \{(x,y)\})\subseteq V_0$, which implies that $N(x)\cap W_2=\varnothing$, but this is a contradiction as $\deg(y)\le \n(H)-2$. Thus, $W_2=W_2'$.

\vspace{.2cm}
\noindent
Subcase 2.3. $\Delta(H)= \n(H)-3$. Assume first that $\delta(H)=0$. By Theorem \ref{teo-UpperBoundGeneral}, for any $\zeta'(G)$-set $S=S_0\cup S_1$ we have $\gamma_R^p(G\circ H)\le 2|S_1|+4|S_0|= \zeta'(G)$.
Now, from Theorem \ref{teo-principal-generalizado} we have that $\gamma_R(G\circ H)=2\gamma_t(G)$, and since $W_2\in \wp_o(G)\cap \mathcal{D}^p(G)$ and $f(H_x)=4$ for every vertex $x\in W_2'$, we have that
$2\gamma_t(G)=\gamma_R(G\circ H)=\gamma_R^p(G\circ H)\le \zeta'(G)\le 2|W_2\setminus W_2'|+4|W_2'|\le \gamma_R(G\circ H)=2\gamma_t(G)$.
Therefore, $\zeta'(G)=2\gamma_t(G)$.

\medskip
On the other side, if $\delta(H)\ge 1$, then we can proceed as  in Subcase 2.2 to deduce that $W_2=W_2'$ is a packing and also a perfect dominating set of $G$, which implies that $\gamma^p(G)=\gamma(G)=\rho(G)$.
Thus,  $f(H_x)=4$ for every vertex $x\in W_2$, and so  by Theorem \ref{teo-principal-generalizado} we have that $2\gamma_t(G)=\gamma_R(G\circ H)=\gamma_R^p(G\circ H)=4|W_2|=4\gamma^p(G)=4\rho(G)$. Therefore, $\gamma_t(G)=2\gamma^p(G)=2\rho(G)$.

\vspace{.2cm}
\noindent
Subcase 2.4. $\Delta(H)\leq \n(H)-4$.
In this case, $|V(H_u)\cap V_1|\geq 3$. Hence,  for any $u'\in N(u)$ and $v\in V(H)$, the function $g'$, defined by $g'(H_u)=g'(H_{u'})=g'(u,v)=g'(u',v)=2$ and $g'(x,y)=f(x,y)$ for the remaining vertices, is an RDF on $G\circ H$ of weight $\omega(g')<\omega(f)=\gamma_R(G\circ H)$, which is  again a contradiction. Therefore, $\Delta(H)\le \n(H)-4$ leads to $V_1= \varnothing$.

\medskip
We proceed to summarize the conclusions derived from the cases above, and to prove the statements.

\vspace{.3cm}
\noindent \textit{Proof of (i)}. Assume $\Delta(H)=\n(H)-1$.  As we have shown in Case 1 and Subcase 2.1, from $\gamma_R^p(G\circ H)=\gamma_R(G\circ H)$ we deduce that  $\mathcal{P}_2(G,H)$ holds.

\medskip
Conversely,  if $\Delta(H)=\n(H)-1$ and $\mathcal{P}_2(G,H)$ holds, then by Theorems \ref{LexicPerfectRomanGraphs}-(i), \ref{teo-char-gammaPerfectLexic} and \ref{teo-principal-generalizado},  $\gamma_R^p(G\circ H)=2\gamma^p(G\circ H)=2\gamma(G)=\gamma_R(G\circ H)$. Therefore, (i) follows.

\vspace{.3cm}
\noindent \textit{Proof of (ii)}. Assume $\Delta(H)=\n(H)-2$.  If $\gamma_R^p(G\circ H)=\gamma_R(G\circ H)$, then we have to consider Case 1 and Subcase 2.2.

\medskip
From Case 1, $\gamma_R(G\circ H)=2\gamma_t(G)$ and $\mathcal{P}_1(G,H)$ holds. Thus, $\delta(H)=0$ and there exists an efficient open dominating set $S$ of $G$ with $|S|=\gamma_t(G)$. Since $(S,\varnothing)$ is a dominating couple and $2|S|=2\gamma_t(G)=\gamma_R(G\circ H)$, by  Theorem~\ref{teo-principal-generalizado} we conclude that
$(S,\varnothing)$ is a $\zeta(G)$-couple and, obviously, $S\in \wp_o(G)$.

On the other hand, in Subcase 2.2 we concluded that $(W_2\setminus W_2',W_2')$ is a $\zeta(G)$-couple and $W_2\in \wp_o(G)$. Also, $W_2=W_2'$ whenever   $\delta(H)\ge 1$.

\medskip
Conversely, let $(A,B)$ be a $\zeta(G)$-couple such that $A\cup B\in \wp_o(G)$.    Let $v\in V(H)$ be a vertex of maximum degree and let $\{v'\}=V(H)\setminus N[v]$.

\medskip
Notice that if $v'$ is an isolated vertex, then the function $g(X_0,X_1,X_2)$, defined by $X_2=A\times \{v'\}\cup B\times \{v\}$ and $X_1=B\times \{v'\}$, is a PRDF on $G\circ H$. Hence, by Theorem~\ref{teo-principal-generalizado},   $\zeta(G)= \gamma_R(G\circ H)\leq\gamma_R^p(G\circ H)\leq \omega(g)= 2|A|+3|B|=\zeta(G)$. Therefore, $\gamma_R^p(G\circ H)=\gamma_R(G\circ H)$.

\medskip
Now, if $\deg(v')\ge 1$ and $A=\varnothing$, then $B$ is a packing and also a dominating set, which implies that the function $g(X_0,X_1,X_2)$, defined by $X_2= B\times \{v\}$ and $X_1=B\times \{v'\}$, is a PRDF on $G\circ H$. Hence, by Theorem~\ref{teo-principal-generalizado},   $\zeta(G)= \gamma_R(G\circ H)\leq\gamma_R^p(G\circ H)\leq \omega(g)= 3|B|=\zeta(G)$. Therefore, $\gamma_R^p(G\circ H)=\gamma_R(G\circ H)$, as required.

\vspace{.3cm}
\noindent \textit{Proof of (iii)}. Assume $\Delta(H)=\n(H)-3$.
If $\gamma_R^p(G\circ H)=\gamma_R(G\circ H)$, then we have to consider Case 1 and Subcase 2.3.

\medskip
In Case 1 we deduced that $\delta(H)=0$,   $\gamma_R^p(G\circ H)=2|W_2|=2\gamma_t(G)$. Furthermore, since $\mathcal{P}_1(G,H)$ holds, $W_2\in \wp_o(G)\cap \mathcal{D}^p(G)$. Thus, $\zeta'(G)\le 2|W_2|$ and, by Theorem~\ref{teo-UpperBoundGeneral},   $2|W_2|=\gamma_R^p(G\circ H)\le \zeta'(G)\le 2|W_2|$, which implies that $\zeta'(G)=2\gamma_t(G)$.

\medskip
In Subcase 2.3, we deduced that if  $\delta(H)=0$, then
$\zeta'(G)=2\gamma_t(G)$, while if $\delta(H)\ge 1$, then
$\gamma_t(G)=2\gamma^p(G)=2\rho(G)$.

\medskip
Conversely, assume that $\zeta'(G)=2\gamma_t(G)$ and $\delta(H)=0$. Let $S=S_0\cup S_1$ be a
$\zeta'(G)$-set. By Theorems \ref{teo-UpperBoundGeneral} and \ref{teo-principal-generalizado}
  we have that $2\gamma_t(G)=\gamma_R(G\circ H)\le \gamma_R^p(G\circ H)\le 2|S_1|+4|S_0|= \zeta'(G)=2\gamma_t(G)$, which implies that $\gamma_R(G\circ H)=\gamma_R^p(G\circ H)$.

\medskip
Now, assume $\delta(H)\ge 1$ and $\gamma_t(G)=2\gamma^p(G)=2\rho(G)$.  Let $v\in V(H)$ be a vertex of maximum degree, $\{v_1,v_2\}=V(H)\setminus N[v]$ and $D$ a $\gamma^p(G)$-set. Notice that the function $g(X_0,X_1,X_2)$, defined by $X_2=D\times \{v\}$ and $X_1=D\times \{v_1,v_2\}$, is a PRDF on $G\circ H$. Hence, by Theorem \ref{teo-principal-generalizado} we have that  $2\gamma_t(G)=\gamma_R(G\circ H)\leq\gamma_R^p(G\circ H)\leq \omega(g)= 2|X_2|+|X_1|=4\gamma^p(G)=2\gamma_t(G)$. Thus, $\gamma_R^p(G\circ H)=\gamma_R(G\circ H)$, which completes the proof of (iii).

\vspace{.3cm}
\noindent \textit{Proof of (iv)}. Assume $\Delta(H)\le \n(H)-4$. As shown in Case 1 and Subcase 2.4, from $\gamma_R^p(G\circ H)=\gamma_R(G\circ H)$ we deduce that  $\mathcal{P}_1(G,H)$ holds.
Conversely,  if $\mathcal{P}_1(G,H)$ holds, then by Theorems \ref{LexicPerfectRomanGraphs}-(iii), \ref{teo-char-gammaPerfectLexic} and \ref{teo-principal-generalizado} it follows that $\gamma_R^p(G\circ H)=2\gamma^p(G\circ H)=2\gamma_t(G)=\gamma_R(G\circ H)$. Therefore, (iv) follows.
\end{proof}

\section{Concluding remarks}

This paper is part of a larger project in which the aim is to propose closed formulae for the domination parameters  of product graphs. In general, these formulae are expressed in terms of various parameters of the graphs involved in the product.
The specific aim of this paper is to  study the case of the perfect domination number,  the Roman domination number and the perfect Roman domination number of lexicographic product graphs.
We show that this goal can be achieved relatively easily for the case of the first two parameters, while for the case of the perfect Roman domination number the picture is completely different.
The impossibility of achieving the target in the case of the latter parameter led us to
obtain  general bounds, and then to give some sufficient and/or necessary conditions for the bounds to be achieved. As a consequence of the results obtained, there are several challenges for future work. Some of them are the traditional ones in domination theory and consist of improving the obtained bounds or trying to characterise the families of graphs that reach them.  In our opinion, a more important challenge is to try to look from another angle to try to achieve the initial objective by trying to find relationships with graph parameters that we have not considered, or with parameters that have not yet been defined and studied.

\providecommand{\noopsort}[1]{}

\end{document}